\documentclass[twocolumn ,aps,prd,a4paper,superscriptaddress,tightenlines,noeprint,10pt,floatfix,nobibnotes,notitlepage]{revtex4-2}
\RequirePackage{float}
\usepackage[normalem]{ulem}
\usepackage{subcaption}
\usepackage[utf8]{inputenc}
\usepackage[USenglish]{babel}
\usepackage[T1]{fontenc}
\usepackage{graphicx,epsf}
\usepackage[toc,page,header]{appendix}
\usepackage{minitoc}
\usepackage{tabularx}
\usepackage{color}
\usepackage[bookmarksnumbered,hypertexnames=false,colorlinks,colorlinks=true,linkcolor=blue,urlcolor=black,citecolor=blue,anchorcolor=green]{hyperref}
\usepackage{bbm,microtype,mathrsfs,amsmath,amssymb,color,amsthm,mathtools,fullpage,enumitem,ae,aecompl,bm,thm-restate}
\usepackage{physics}
\usepackage{tikz}
\usepackage{lipsum}
\usepackage{nameref}
\usepackage{color}
\usepackage{booktabs, caption, palatino}
\usepackage[capitalize]{cleveref}
\usepackage{silence}

\usepackage{mathtools}
\usepackage{braket}

\allowdisplaybreaks
\newcommand{\be}{\begin{equation}}
\newcommand{\ee}{\end{equation}}
\newcommand{\ba}{\begin{eqnarray}}
\newcommand{\ea}{\end{eqnarray}}

\newcommand{\half}{\frac{1}{2}}

\newcommand{\ignore}[1]{}

\newcommand{\ot}{\otimes}

\newcommand{\pauli}[1]{\mathbb{P}_{#1}}

\newcommand{\de}[0]{{\operatorname{d}}}

\def\CC{{\rm\kern.24em \vrule width.04em height1.46ex depth-.07ex
   \kern-.29em C}}
\def\P{{\rm I\kern-.25em P}}
\def\RR{{\rm
        \vrule width.04em height1.58ex depth-.0ex
        \kern-.04em R}}

\def\bbbc{{\mathchoice {\setbox0=\hbox{$\displaystyle\rm C$}\hbox{\hbox
to0pt{\kern0.4\wd0\vrule height0.9\ht0\hss}\box0}}
{\setbox0=\hbox{$\textstyle\rm C$}\hbox{\hbox
to0pt{\kern0.4\wd0\vrule height0.9\ht0\hss}\box0}}
{\setbox0=\hbox{$\scriptstyle\rm C$}\hbox{\hbox
to0pt{\kern0.4\wd0\vrule height0.9\ht0\hss}\box0}}
{\setbox0=\hbox{$\scriptscriptstyle\rm C$}\hbox{\hbox
to0pt{\kern0.4\wd0\vrule height0.9\ht0\hss}\box0}}}}
\def\bbbz{{\mathchoice {\hbox{$\sf\textstyle Z\kern-0.4em Z$}}
{\hbox{$\sf\textstyle Z\kern-0.4em Z$}}
{\hbox{$\sf\scriptstyle Z\kern-0.3em Z$}}
{\hbox{$\sf\scriptscriptstyle Z\kern-0.2em Z$}}}}

\newlength{\fighskip} \fighskip=2pt
\newlength{\figvskip} \figvskip=1pt

\makeatletter
\def\namedlabel#1#2{\begingroup
   \def\@currentlabel{#2}%
   \label{#1}\endgroup
}
\makeatother
\usepackage{hyperref}
\newtheorem{theorem}{Theorem}

\newtheorem{lemma}{Lemma}
\newtheorem{definition}{Definition}

\definecolor{Large_Ellipse_Env}{HTML}{ECD9ED}
\definecolor{Small_Ellipse_Env}{HTML}{ECD9ED}
\definecolor{Large_Ellipse_Prb}{HTML}{FFB98D}
\definecolor{Small_Ellipse_Prb}{HTML}{FFE2A1}

\begin{document}
\setcounter{secnumdepth}{3}

\title{Induced Resource Theories and Harvesting via Quantum Probes}

\author{Ron Nystr\"om}
\affiliation{Department of Physics, P.O.Box 64, FIN-00014 University of Helsinki, Finland}

\author{Simone Cepollaro}
\affiliation{Scuola Superiore Meridionale, Largo S. Marcellino 10, 80138 Napoli, Italy}
\affiliation{INFN, Sezione di Napoli, Italy}

\author{Nicola Pranzini}
\affiliation{Department of Physics, P.O.Box 64, FIN-00014 University of Helsinki, Finland}
\affiliation{QTF Centre of Excellence, Department of Physics,
University of Helsinki, P.O. Box 43, FI-00014 Helsinki, Finland}
\affiliation{InstituteQ - the Finnish Quantum Institute, Finland}

\author{Stefano Cusumano}
\affiliation{INFN, Sezione di Napoli, Italy}
\affiliation{Dipartimento di Fisica `Ettore Pancini', Universit\`a degli Studi di Napoli Federico II,
Via Cintia 80126,  Napoli, Italy}

\author{Alioscia Hamma}
\affiliation{Scuola Superiore Meridionale, , Largo S. Marcellino 10, 80138 Napoli, Italy}
\affiliation{INFN, Sezione di Napoli, Italy}
\affiliation{Dipartimento di Fisica `Ettore Pancini', Universit\`a degli Studi di Napoli Federico II, Via Cintia 80126,  Napoli, Italy}

\author{Esko Keski-Vakkuri}
\affiliation{Department of Physics, P.O.Box 64, FIN-00014 University of Helsinki, Finland}
\affiliation{InstituteQ - the Finnish Quantum Institute, Finland}
\affiliation{Helsinki Institute of Physics, P.O.Box 64, FIN-00014 University of Helsinki, Finland}

\begin{abstract}
We consider scenarios in which a quantum system with a well-defined resource theory is used as a probe to interact with an environment, such as a quantum field, for which a resource-theoretic description is absent or incomplete. We clarify if and how the harvesting of a resource in the probe can tell us about the state of the environment. This is particularly ambiguous when the probe–environment interaction is not a free operation, or the concept of such free operations cannot be defined altogether. We propose a framework and precise conditions under which it becomes possible to interpret resource generation on the probe as evidence of resources in the environment, thereby introducing an effective notion of resources for the latter. Our results clarify in which sense resources can be said to be harvested from the environment and provide a systematic way to analyse such processes beyond fully controlled resource-theoretic settings. More generally, this work may provide a step towards a more general understanding of the interplay of different quantum resources. 
\end{abstract}

\maketitle
\section{Introduction}
\label{sec:introduction}

Quantum resource theories (QRTs) provide a framework to characterise and quantify non-classical features of quantum systems in terms of free operations and free states~\cite{ChitambarG19}. This has allowed the rigorous analysis of several quantum resources, ranging from quantum entanglement~\cite{RevModPhys.81.865} and coherence~\cite{PhysRevLett.116.120404,RevModPhys.89.041003}, to non-locality~\cite{de_vicente_nonlocality_2014}, non-stabilizerness~\cite{bravyi2005universal, veitch_resource_2014,PhysRevLett.128.050402}, anti-flatness~\cite{Tirrito2024, Jasser2026}, asymmetry~\cite{marvian_coherence_2020,PhysRevLett.129.190502,Zhou2021newperspectives} and thermodynamics~\cite{lostaglio_introductory_2019}. However, while well-established for finite-dimensional systems, the extension of these general notions to more complex settings, such as quantum fields, remains largely unexplored.

A prominent context where this issue arises is that of resource harvesting protocols~\cite{valentini1991non,reznik2003entanglement}. In these, localised 
probes interact with a quantum field and acquire non-classical properties such as entanglement~\cite{PhysRevD.103.016007,martin-martinez_entanglement_2014,PhysRevD.102.125026,PhysRevD.106.076002,Pozas-KerstjensEtAl15,PercheEtAl24}, non-stabilizerness and non-locality~\cite{9ph7-cyzh,4brj-cl26, YangEtAl2025, ZhangEtAl26} as well as contextuality~\cite{LeMaitre2026, LimaEtAl2025}. These phenomena are often interpreted as evidence that the corresponding resource was already present in the field~\cite{SUMMERS1985257,summers_maximal_1987} and has been ``harvested'' by the probes.

However, this interpretation is not always straightforward~\cite{PhysRevD.104.125005}. Besides the possible aforementioned absence of a resource theory of the corresponding resource for a quantum field, one could also find other constraints hindering the proper definition of a harvesting protocol. For instance, the probe–field interaction, constrained by some physical requirements, might not be a free operation. Furthermore, one might face the problem of interpolating between different resource theories, where a free state of a given resource theory might not be free in another~\cite{DenerisEtAl26}. As a consequence, the generation of a resource in the probe in these scenarios does not, by itself, imply the presence of that resource in the field. This raises a fundamental question: in what sense can resources be said to be harvested from an environment that lacks a well-defined resource theory?

In this work, we address this question by proposing a framework in which the resource theory of a probe system induces an effective notion of resources for the environment. Rather than assuming \textit{a priori} a resource theory for the environment, we infer it operationally through its action on the probe, mediated by some given interaction. We present a system-independent construction of such an induced resource theory, based on employing the well-defined resource theory on the probe to establish a classification of environment states which are free from the perspective of the probe, i.e., environment states that do not generate the resource in the probe via the given interaction.

Within this framework, we analyse different harvesting scenarios, distinguished by the level of control and prior knowledge about the interaction and the environment. This classification clarifies under which assumptions resource generation on the probe can be consistently interpreted as harvesting from the environment.

The manuscript is structured as follows. In Sec.~\ref{sec:induced_QRT} we introduce the notion of induced resource theory, that is, resource theories constrained upon given unitary operations and set of states, providing examples and ideas from which those definitions are inspired. Then in Sec.~\ref{sec:examples} we provide several examples of the notions introduced in Sec.~\ref{sec:induced_QRT}. Finally, in Sec.~\ref{sec:conclusions} we draw our conclusions and provide some outlook for further investigations.

\section{Induced QRT and harvesting}
\label{sec:induced_QRT}

We consider a composite system $\Psi = \mathcal{P}\cup \mathcal{E}$ consisting of a probe system $\mathcal{P}$ and an environment $\mathcal{E}$. The subsystems are initially prepared in states $\rho \in \mathcal{D}(\mathcal{H}_\mathcal{P})$ and $\sigma \in \mathcal{D}(\mathcal{H}_\mathcal{E})$, respectively, and allowed to interact via a unitary $V \in \mathcal{U}(\mathcal{H}_\Psi)$. This induces the channels
\begin{align}
    \Phi_\mathcal{P}^{(\sigma,V)}(\rho) &= \Tr_\mathcal{E}\!\left[V(\rho\otimes\sigma)V^\dag\right], \label{e.direct_channel} \\
    \Phi_\mathcal{E}^{(\rho,V)}(\sigma) &= \Tr_\mathcal{P}\!\left[V(\rho\otimes\sigma)V^\dag\right], \label{e.complementary_channel}
\end{align}
describing, respectively, the effect of the environment on the probe and the back-action of the probe on the environment. For simplicity, and without loss of generality we restrict to unitary interactions.

We assume that the probe $\mathcal{P}$ is equipped with a resource theory specified by a set of free states $\mathcal{F}_\mathcal{P} \subset \mathcal{D}(\mathcal{H}_\mathcal{P})$. These implicitly define a set of free operations $\mathcal{O}_{\cal P}$ as those operations mapping free states into free states. There may also exist resource theories defined for the environment $\mathcal{E}$ and the composite system $\Psi$. However, these can be difficult to characterise in general. This is particularly true for composite systems, where the resource theories of the probe and environment behave fundamentally differently, and therefore may not admit a well-defined unified description. Hence, we assume no \textit{a priori} resource theory for the environment $\mathcal{E}$ nor for the composite system $\Psi$. Our goal is to use the probe's resource theory and the interaction to induce an effective notion of resources on the environment.

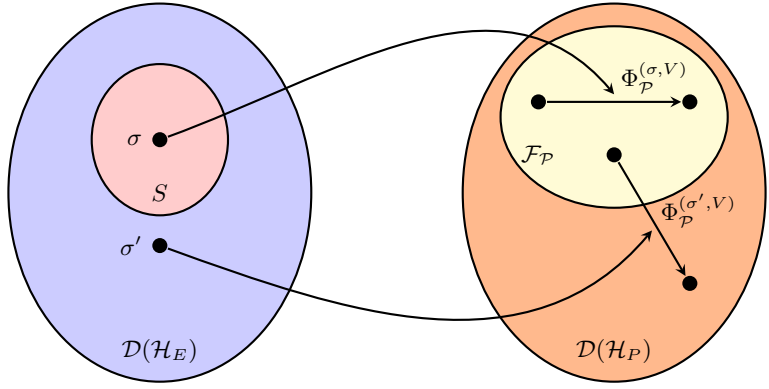
\begin{figure*}[!ht]
    \centering
    \begin{tikzpicture}[>=stealth, thick]

% Left set
\draw[fill=blue!20] (-3,0) ellipse (2 and 2.5);
\node at (-3,-2.1) {$\mathcal{D}(\mathcal H_E)$};
%free set in E
\draw[fill=red!20] (-3,0.7) ellipse (0.9 and 1);
\node at (-3,0) {$S$};

\node[circle, fill=black, inner sep=2pt, label=left:{$\sigma$}] (e1) at (-3,0.7) {};
\node[circle, fill=black, inner sep=2pt, label=left:{$\sigma'$}] (e2) at (-3,-0.7) {};

% Right set
\draw[fill=Large_Ellipse_Prb] (3,0) ellipse (2 and 2.5);
\node at (3,-2.1) {$\mathcal{D}(\mathcal H_P)$};
%free set in P
\draw[fill=yellow!20] (3,1) ellipse (1.5 and 1.2);
\node at (2,0.5) {$\mathcal{F_P}$};

\node[circle, fill=black, inner sep=2pt] (p1) at (2,1.2) {};
\node[circle, fill=black, inner sep=2pt] (p2) at (4,1.2) {};
\node[circle, fill=black, inner sep=2pt] (p3) at (3,0.5) {};
\node[circle, fill=black, inner sep=2pt] (p4) at (4,-1.2) {};

% Pairwise arrows inside right set
\draw[->] (p1) -- (p2) node[pos=0.8, above] {$\Phi^{(\sigma,V)}_\mathcal{P}$};
\draw[->] (p3) -- (p4) node[pos=0.6, above, xshift=15pt, yshift=-2pt] {$\Phi^{(\sigma',V)}_\mathcal{P}$};

% Midpoints of arrows (targets)
\coordinate (m1) at (3,1.3);
\coordinate (m2) at (3.5,-0.5);

% Arrows from left points to arrows on the right
\draw[->] (e1) to[out=20, in=130] (m1);   % bends upward
\draw[->] (e2) to[out=-20, in=-130] (m2); % bends downward

\end{tikzpicture}
    \caption{Induced set of free states $S$: for a given $V$, the set $S\subseteq\mathcal{D(H_E)}$ contains states $\sigma$ for which $\Phi^{(\sigma,V)}_\mathcal{P}$ is a free operation of $(\mathcal{F_P},\mathcal{O_P})$.}
    \label{f.weak_induced_QRT}
\end{figure*}

\subsection{Weak and strong induced resource theories}
\label{s.w/s_QRT}

We begin by identifying those states of the environment that cannot generate resource in the probe.

\begin{definition} \label{d.S}
Let $V$ be a unitary operation describing the interaction between $\mathcal{P}$ and $\mathcal{E}$, and $(\mathcal{F}_\mathcal{P},\mathcal{O}_\mathcal{P})$ a resource theory on $\mathcal{P}$. A set $S \subseteq \mathcal{D}(\mathcal{H}_\mathcal{E})$ is called an \emph{induced set of free states} if
\begin{equation}
    \Phi_\mathcal{P}^{(\sigma,V)}(\rho) \in \mathcal{F}_\mathcal{P}
\end{equation}
for all $\sigma \in S$ and all $\rho \in \mathcal{F}_\mathcal{P}$.
\end{definition}

In other words, $S$ consists of environment states that, when interacting with any free probe state according to the unitary operation $V$, cannot generate the resource in the probe. This induces a notion of ``free'' states for the environment, defined relative to the interaction $V$ and the probe resource theory. At this point, the candidate induced free set $S$ and the interaction $V$ can be chosen quite generally and could be motivated by experimental or theoretical reasons. Later, for a given interaction $V$ we will be interested in what is the maximal induced free set, and denote it $S_\textrm{max}$.

Given such a set $S$, one may consider a set of quantum operations that preserve it~\footnote{Among all possible such operations, one usually selects only a subset, either due to physical constraints~\cite{ChitambarG19} or mathematical ones~\cite{DiazEtAl2025}}:
\begin{equation}
{\rm Self}(S) = 
\{\Phi: S \to \mathcal{D}(\mathcal{H}_\mathcal{E}) \;|\; \Phi(\sigma) \in S\}.
\end{equation}

\begin{definition} \label{d.weak_iQRT}
The triple $(S,V,\mathcal{F}_\mathcal{P})$ defines a \emph{weak induced resource theory} on $\mathcal{E}$ given by the pair $(S, {\rm Self}(S))$.
\end{definition}

The qualifier ``weak'' reflects the fact that no condition is imposed on the evolution of the environment under the interaction: only the effect on the probe is used to define $S$. Furthermore, the definition is agnostic about whether $\mathrm{Self}(S)$ turns out to be trivial (containing only the identity map).

A stronger notion is obtained by additionally requiring stability of $S$ under the back-action of the probe.

\begin{definition} \label{d.strong_iQRT}
The triple $(S,V,\mathcal{F}_\mathcal{P})$ defines a \emph{strong induced resource theory} if it satisfies Definition~\ref{d.weak_iQRT} and, in addition,
\begin{equation}
    \Phi_\mathcal{E}^{(\rho,V)}(\sigma) \in S
\end{equation}
for all $\sigma \in S$ and all $\rho \in \mathcal{F}_\mathcal{P}$.
\end{definition}

In this case, the interaction with free probe states preserves the set $S$, i.e., $\Phi_\mathcal{E}^{(\rho ,V)}\in \textrm{Self}(S) \ \forall \rho \in \mathcal{F}_\mathcal{P}$. In particular, this implies that the induced free states are stable under repeated probe–environment interactions.

The distinction between weak and strong induced resource theories becomes especially relevant when considering multiple interactions with the same environment. In the strong case, repeated interactions with fresh probes in free states cannot generate resource in the probe if the initial environment state lies in $S$. In contrast, this stability is not guaranteed for weak induced resource theories: resource generation may occur after one or many subsequent interactions even if it does not occur in a single step. This situation is relevant, for instance, when probing a large or uncontrollable environment through successive interactions with fresh probes, as in collision models~\cite{CiccarelloEtAl22,e24091258,CattaneoEtAl22}. In typical harvesting scenarios involving quantum fields, however, one considers single probe–environment interactions, and the weak notion is often sufficient.

\subsection{Operational interpretation of harvesting}
\label{s.Harvesting_classification}

We now revisit the notion of resource harvesting in the light of the induced resource theories introduced above. In general, harvesting refers to the generation of resources in the probe after the interaction with the environment, i.e.,
\begin{equation}
    r_\mathcal{P}\!\left( \Phi_\mathcal{P}^{(\sigma,V)}(\rho) \right) > 0,
\end{equation}
for some initial probe state $\rho \in \mathcal{F}_\mathcal{P}$, environment state $\sigma \in \mathcal{D}(\mathcal{H}_\mathcal{E})$ and a resource monotone $r_\mathcal{P}:\mathcal{D}(\mathcal{H}_\mathcal{P})\rightarrow\mathbb{R}_+$ for the resource theory of the probe. Fig.~\ref{f.weak_induced_QRT} illustrates how the induced resource theory classifies which environment states facilitate harvesting of a resource in the probe. 

In the case of a strong induced resource theory, defined by the tuple $(S,V,\mathcal{F}_\mathcal{P})$, the set $S$ is stable under the back-action of the probe, and repeated interactions with free probe states cannot generate resource in the probe if the environment is initially in $S$. In this case, observing resource generation in the probe implies that the initial state of the environment does not belong to $S$, and harvesting can be interpreted as witnessing environmental resources relative to this induced structure (see Sec.~\ref{s.H_as_W}).

If instead only a weak induced resource theory is available, the set $S$ is defined solely through the effect of the environment on the probe, without imposing stability under the interaction. In this case, observing resource generation in a single interaction still implies that the corresponding state $\sigma$ does not belong to $S$. However, this conclusion is not stable under repeated interactions: resource generation may occur after multiple interactions with free probe states even if it does not occur in a single step. As a result, the interpretation of harvesting as witnessing environmental resources is more limited and depends on the specific probing protocol.

Finally, depending on the choice of the interaction $V$, the presence of resource generation in the probe does not necessarily imply that the environment itself possesses a non-zero amount of that resource. In such cases, harvesting becomes a purely phenomenological feature of the probe–environment dynamics, lacking a direct interpretation in resource-theoretic terms for the environment.

\subsection{Induced free states and harvesting as a witness}
\label{s.H_as_W}
Given a probe resource theory $(\mathcal{F}_\mathcal{P},\mathcal{O}_\mathcal{P})$ and an interaction $V$, a natural object is the maximal set of environmental states that do not generate resource in the probe:

\begin{definition}
The  maximal induced set of free states is
\begin{equation}
S_{\max} := \left\{ \sigma \in \mathcal{D}(\mathcal{H}_\mathcal{E}) ~\middle|~ 
\Phi^{(\sigma,V)}_\mathcal{P}(\rho) \in \mathcal{F}_\mathcal{P}
\ \forall \rho \in \mathcal{F}_\mathcal{P}
\right\}.
\end{equation}
\end{definition}

By construction, $S_{\max}$ depends on $\mathcal{F_P}$ and $V$, and contains all environmental states that are \emph{indistinguishable from free states} when probed through $V$ using free probe states. In particular, any induced resource theory $(S,V,\mathcal{F}_\mathcal{P})$ satisfies $S \subseteq S_{\max}$. Additionally, $S_\text{max}$ inherits some properties of $\mathcal{F_P}$. Namely, if $\mathcal{F_P}$ is convex, then $S_\text{max}$ is also convex~\footnote{The proof consists of considering two states $\sigma,\sigma' \in S_\text{max}$ and a generic linear combination $\xi:=p\sigma+(1-p)\sigma'$ with $p \in [0,1]$. One has for $\rho \in \mathcal{F_P}$ that $\Phi_\mathcal{P}^{(\xi,V)}(\rho) = p \Phi_\mathcal{P}^{(\sigma,V)}(\rho) + (1-p) \Phi_\mathcal{P}^{(\sigma',V)}(\rho)$. Since $\sigma,\sigma' \in S_\text{max}$, it follows that $\Phi_\mathcal{P}^{(\sigma,V)}(\rho), \Phi_\mathcal{P}^{(\sigma',V)}(\rho) \in \mathcal{F_P}$. If $\mathcal{F_P}$ is convex, then $\Phi_\mathcal{P}^{(\xi,V)}(\rho) \in \mathcal{F_P}$. Thus, $\xi \in S_\text{max}$ and the entire set must be convex. Note that convexity of $S_\text{max}$ does not imply that the complete induced resource theory $(S_\text{max},\text{Self}(S_\text{max}))$ is convex~\cite{ChitambarG19}.}, which is useful when considering convex QRTs, such as coherence and entanglement~\cite{ChitambarG19, RegulaB2018, BrandaoEtAl2015, NaseriEtAl26}.

The operational relevance of $S_{\max}$ is immediate: whenever an initially free probe acquires resource after the interaction, the environmental state must lie outside $S_{\max}$.

\begin{lemma}
If there exist $\rho \in \mathcal{F}_\mathcal{P}$ such that
\begin{equation}
r_\mathcal{P}\!\left(\Phi^{(\sigma,V)}_\mathcal{P}(\rho)\right) > 0,
\end{equation}
then $\sigma \notin S_{\max}$.
\end{lemma}

While this statement follows directly from the definition of $S_{\max}$, it highlights its operational meaning: $S_{\max}$ characterises precisely those states of the environment that cannot be detected as resourceful through the chosen interaction and probe.

This observation becomes non-trivial when a resource theory $(\mathcal{F}_\mathcal{E},\mathcal{O}_\mathcal{E})$ on the environment is available. In that case, comparing $\mathcal{F}_\mathcal{E}$ with $S_{\max}$ determines whether harvesting can be used as a \emph{resource witness}.

\begin{theorem} \label{t.FE_subset_S}
If $\mathcal{F}_\mathcal{E} \subseteq S_{\max}$, then any generation of resource in the probe implies that the environmental state is resourceful, i.e.,
\begin{equation}
r_\mathcal{P}\!\left(\Phi^{(\sigma,V)}_\mathcal{P}(\rho)\right) > 0
\quad \Rightarrow \quad
\sigma \notin \mathcal{F}_\mathcal{E}.
\end{equation}
\end{theorem}

\begin{proof}
If $\sigma \in \mathcal{F}_\mathcal{E}$, then by assumption $\sigma \in S_{\max}$, which implies $\Phi^{(\sigma,V)}_\mathcal{P}(\rho) \in \mathcal{F}_\mathcal{P}$ for all $\rho \in \mathcal{F}_\mathcal{P}$. The claim follows by contraposition.
\end{proof}

Therefore, under the condition $\mathcal{F}_\mathcal{E} \subseteq S_{\max}$, the probe acts as a detector of resources in the environment. Conversely, if this inclusion fails, harvesting does not in general provide reliable information about whether $\sigma$ is resourceful in its own resource theory. This Theorem is illustrated in Fig.~\ref{f.harvesting_max}.

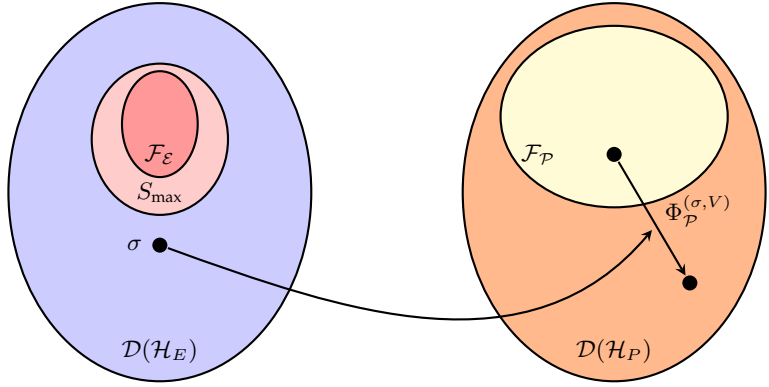
\begin{figure*}[!ht]
    \centering
    \begin{tikzpicture}[>=stealth, thick]

% Left set
\draw[fill=blue!20] (-3,0) ellipse (2 and 2.5);
\node at (-3,-2.1) {$\mathcal{D}(\mathcal H_E)$};

%free set S
\draw[fill=red!20] (-3,0.7) ellipse (0.9 and 1);
\node at (-3,0) {$S_\text{max}$};

%free set in E
\draw[fill=red!40] (-3,0.9) ellipse (0.5 and 0.7);
\node at (-3,0.5) {$\mathcal{F_E}$};

\node[circle, fill=black, inner sep=2pt, label=left:{$\sigma$}] (e2) at (-3,-0.7) {};

% Right set
\draw[fill=Large_Ellipse_Prb] (3,0) ellipse (2 and 2.5);
\node at (3,-2.1) {$\mathcal{D}(\mathcal H_P)$};
%free set in P
\draw[fill=yellow!20] (3,1) ellipse (1.5 and 1.2);
\node at (2,0.5) {$\mathcal{F_P}$};

\node[circle, fill=black, inner sep=2pt] (p3) at (3,0.5) {};
\node[circle, fill=black, inner sep=2pt] (p4) at (4,-1.2) {};

% Pairwise arrows inside right set
\draw[->] (p3) -- (p4) node[pos=0.6, above, xshift=15pt, yshift=-2pt] {$\Phi^{(\sigma,V)}_\mathcal{P}$};

% Midpoints of arrows (targets)
\coordinate (m1) at (3,1.3);
\coordinate (m2) at (3.5,-0.5);

% Arrows from left points to arrows on the right
\draw[<-] (m2) to[out=-130, in=-20] (e2); % bends downward

\end{tikzpicture}
    \caption{Depiction of Theorem~\ref{t.FE_subset_S}. If there is non-zero resource in the initially free probe state after interaction, then the environment state is $\sigma \in \mathcal{D}\setminus S_{\max}$. Additionally, if $\mathcal{F_E}\subseteq S$, then the state of the environment is resourceful, i.e. $\sigma \notin \mathcal{F_E}$. Not seeing the resource in $\mathcal{P}$ does not imply the environment was in a free state, unless $\mathcal{F_E}=S_\text{max}$.}
    \label{f.harvesting_max}
\end{figure*}

A particularly transparent situation arises when $\mathcal{F}_\mathcal{E} = S_{\max}$, in which case the probe perfectly detects the presence of the resource:
\begin{equation}
r_\mathcal{P}\!\left(\Phi^{(\sigma,V)}_\mathcal{P}(\rho)\right) > 0
\quad \Leftrightarrow \quad
\sigma \notin \mathcal{F}_\mathcal{E}.
\end{equation}

The concept of induced free states depends on a chosen interaction $V$, while the set of free states $\mathcal{F}_\mathcal{E}$ is defined independently of probes and interactions. Therefore, it is expected that $S_{\max}$ can also be strictly smaller than $\mathcal{F}_\mathcal{E}$. This limits the applicability of harvesting as a witness. This limitation can sometimes be alleviated by restricting to a physically relevant subset of states of the environment $\mathcal{C} \subsetneq \mathcal{D(H_E)}$ (e.g., thermal states or states satisfying additional constraints), for which the inclusion $\mathcal{F}_\mathcal{E} \cap \mathcal{C} \subseteq S_{\max}$ may hold.

Finally, a more practical construction can be obtained by fixing the initial probe state. Given $\rho_0 \in \mathcal{F}_\mathcal{P}$, define
\begin{equation}
S_{\max}^{(\rho_0)} := \left\{ \sigma \in \mathcal{D}(\mathcal{H}_\mathcal{E}) ~\middle|~ 
\Phi^{(\sigma,V)}_\mathcal{P}(\rho_0) \in \mathcal{F}_\mathcal{P} \right\}.
\end{equation}
This set characterises those environmental states that appear free when probed with the specific input $\rho_0$, and can be used to construct experimentally accessible witness schemes under analogous inclusion conditions.

\section{Examples}
\label{sec:examples}
\subsection{Qubit--qubit controlled Clifford}

We begin with a simple setting in which both the probe and the environment are qubits, allowing for a full and explicit analysis. This example serves to illustrate how induced resource theories arise and how their structure depends on the interaction.

We take the probe resource to be the stabilizer Rényi entropy, with free states $\mathcal{F}_\mathcal{P} = \rm{STAB}_0$, i.e., the set of states with vanishing stabilizer Rényi entropy \cite{leone2022stabilizer} defined as $M_{\alpha}(\ket{\psi})=(1-\alpha)^{-1}\log\sum_{P\in\mathbb{P}_n}\Xi_P^\alpha(\ket{\psi})-n -S_\alpha (\psi)$ where $\Xi_P(\ket{\psi})=2^{-n} {\bra{\psi}P\ket{\psi}^2}$ is the Pauli spectrum for a $n$-qubit system  Pauli operators $P$ and $S_\alpha$ is the usual R\'enyi entropy. The free set $\rm{STAB}_0$ is the set of states that do not admit a purification into stabilizer states, and these are the on the orbit of the Clifford group through any computational basis state. For a single qubit, this corresponds to pure stabilizer states together with the maximally mixed state. A more general discussion of stabilizer Rényi entropy is included in App.~\ref{s.SRE}.

As interaction, we consider a controlled Clifford unitary
\begin{equation}
\label{e.cnot_clean}
V = \mathbb{I} \otimes \Pi_0 + R \otimes \Pi_1,
\end{equation}
where $R$ is a Clifford operator acting on the probe and $\Pi_k = \dyad{k}$ acts on the environment. The environment thus controls whether the probe undergoes the identity or a Clifford transformation.

For a generic environmental state $\sigma$ and $\rho \in \rm{STAB_0}$, the induced channel on the probe reads
\begin{equation}
\label{e.rho_CV_clean}
\Phi_\mathcal{P}^{(\sigma,V)}(\rho)
= \Tr[\Pi_0 \sigma]\, \rho
+ \Tr[\Pi_1 \sigma]\, R \rho R.
\end{equation}
That is, the probe undergoes a mixture of identity and a Clifford operation.

Even though both the identity and any Clifford unitary preserve stabilizer states, the  above map  does not for every $\sigma$. However, for $\sigma$  in $\rm{STAB}_0$, the map $\Phi_\mathcal{P}^{(\sigma,V)}$ preserves $\rm{STAB}_0$ so we see that we induced theory is  $S=\rm{STAB_0}$. In this example, the induced theory is the stabilizer resource theory on the first qubit, and the example reduces to a consistency check.

We now examine the induced channel on the environment,
\begin{equation}
\Phi_\mathcal{E}^{(\rho,V)}(\sigma)
= \Tr_\mathcal{P}\!\left[V(\rho \otimes \sigma)V^\dagger\right].
\end{equation}
One can observe that, from the same considerations done for $\Phi_{\cal P}^{(\sigma,V)}$, this channel also preserves the set ${\rm STAB_0}\subset \mathcal{D(H_E)}$ for all $\rho\in{\rm STAB_0}\subset \mathcal{D(H_P)}$, and hence the induced resource theory is strong. Thus, this example illustrates how the structure of the interaction determines the induced free set: here, the state of the environment effectively selects between free operations on the probe, resulting in an induced resource theory that is stable under repeated interactions.

We can further analyse this setting from the perspective of resource detection. From Eq.~\eqref{e.rho_CV_clean}, the probe remains in a free state for all $\rho \in \mathrm{STAB}_0$ if and only if the environmental state $\sigma$ induces an equally weighted combination of Clifford operations. In this example, this condition yields the maximal set of induced free states 
\begin{equation}
    S_\text{max} = \left\{\sigma \in \mathcal{D(H_E)} \mid \Tr[\Pi_0 \sigma] \in \{0,\half,1\} \right\}~.
\end{equation}
For the case of a measurement in the computational basis, this corresponds to states of the environment being on the poles and $XY$-plane of the Bloch sphere. This shows once again that the maximal set $S_{\rm max}$ is strongly influenced by the choice of the measurement basis, i.e. of the unitary $V$.

Since the environment is also a qubit, we may endow it with the same resource theory of non-stabilizerness, so that $\mathcal{F}_\mathcal{E} = \mathrm{STAB}_0$. One then finds that $\mathcal{F}_\mathcal{E} \subseteq S_{\max}$, which implies that any generation of resource in the probe certifies that the environmental state is resourceful. In this sense, the interaction $V$ enables harvesting in the probe to act as a witness of non-stabilizerness in the environment.

Notice that this behaviour depends crucially on the choice of interaction. For instance, if one replaces the Clifford operator $R$ with a non-Clifford unitary, such as $R = e^{i\pi X/4}$, the maximal induced set $S_\text{max}$ reduces to a much smaller set (i.e., only $\dyad{0}$). In this case, one has $S_\text{max} \subset \mathcal{F_E}=\rm{STAB_0}$, and the probe can acquire resource even when the environment is initially in a free state. Consequently, the harvesting protocol no longer provides a reliable witness of the resource in the environment. This is however quite natural, since we have chosen an interaction which is not free for the resource theory of stabilizer entropy defined on the combined two-qubit system.

\subsection{Control-flat interactions}

We now present an example where (i) we deal with general finite-dimensional systems and (ii) there is no previously defined resource theory on the environment. 
 Let the environment $\mathcal{E}$ have Hilbert space $\mathcal{H}_\mathcal{E}\cong\mathbb{C}^{d_\mathcal{E}}$, equipped with a resolution of the identity $\{\Pi_k\}_{k=0}^{n-1}$, with $n\le d_\mathcal{E}$ and $\sum_k \Pi_k=\mathbb{I}$. The probe $\mathcal{P}$ is a qudit with $\mathcal{H}_\mathcal{P}=\mathbb{C}^{d_\mathcal{P}}$, and we fix $\mathcal{F}_\mathcal{P}=\mathrm{STAB}_0$.

We consider the interaction
\begin{equation}
    V=\sum_{k=0}^{n-1} X^k \otimes \Pi_k,
\end{equation}
where $X=\sum_{j=0}^{d_\mathcal{P}-1}\dyad{j+1}{j}$ is the shift operator. This is a direct generalization of the controlled NOT gate.

Given a state $\sigma$ of the environment, the induced channel on the probe reads
\begin{equation}
    \Phi_\mathcal{P}^{(\sigma,V)}(\rho)
    = \sum_k \Tr[\Pi_k\sigma]\, X^k \rho X^{k\dagger}.
\end{equation}
Thus, the probe undergoes a convex mixture of shifts, with weights determined by the populations of $\sigma$ in the sectors defined by $\{\Pi_k\}$.

In view of the single qubit example, it is natural to consider the set of flat states:
\begin{equation}
    S=\left\{\sigma~\middle|~\Tr[\Pi_k\sigma]=\frac{1}{n}\ \forall k\right\}~,
\end{equation}
characterised by a flat entanglement spectrum. For $\sigma\in S$, the induced map is the uniform mixture of Clifford unitaries,
\begin{equation}
    \Phi_\mathcal{P}^{(\sigma,V)}(\rho)=\frac{1}{n}\sum_k X^k\rho X^{k\dagger},
\end{equation}
which preserves $\mathrm{STAB}_0$. This means that the interaction $V$ and $\mathrm{STAB}_0$ resource theory on the probe induces the set of flat states as  a weak  resource theory on the environment. Recently, the resource of non-flat states has been put forward and studied for its connection with other quantum resources \cite{PhysRevD.99.066012, Jasser2026} and this example shows how it can be induced by the stabilizer resource theory.  Moreover, a direct computation shows that $S$ is also invariant under the back-action on $\mathcal{E}$, so the induced resource theory is in fact strong.

More generally, the maximal induced free set is determined by the requirement that the above convex combination remains an equally weighted mixture of Clifford operations. This holds precisely when the non-zero weights are all equal. If $\sigma$ has support on $m$ sectors, this implies
\begin{equation}
    \Tr[\Pi_k\sigma]\in\{0,1/m\},
\end{equation}
and therefore
\begin{equation}
    S_{\max}=\left\{\sigma~\middle|~\Tr[\Pi_k\sigma]\in\{0,1/m\}\ \text{for some } m\right\}.
\end{equation}

From the harvesting perspective, this structure has a simple interpretation: resource generation in the probe occurs whenever $\sigma\notin S_{\max}$, i.e. whenever the weights across the sectors are not flat. Thus, the probe detects deviations from uniformity in the decomposition of $\sigma$ induced by $\{\Pi_k\}$.

If we endow $\mathcal{E}$ with the resource theory of non-stabilizerness, $\mathcal{F}_\mathcal{E}=\mathrm{STAB}_0$, the relation between $\mathcal{F}_\mathcal{E}$ and $S_{\max}$ depends on the choice of projectors. For instance, if $\{\Pi_k\}$ project onto computational basis states, then balanced mixtures of these basis states are stabilizer states, and one finds $\mathcal{F}_\mathcal{E}\subseteq S_{\max}$. In this regime, resource generation in the probe certifies non-stabilizerness in the environment, in agreement with the general witnessing mechanism discussed earlier.

\subsection{Qubit--oscillator E-CNOT} \label{s.QHO_ECNOT}

Let us now consider a qubit probe coupled to a quantum harmonic oscillator with Hamiltonian
\begin{equation}
    H=\omega\left(a^\dagger a + \tfrac{1}{2}\right),
\end{equation}
and eigenstates $\{\ket{n}\}$. The interaction is defined by an energy-threshold controlled unitary,
\begin{equation}
    V=\mathbb{I}\otimes \Pi_- + X\otimes \Pi_+,
\end{equation}
where
\begin{equation}
    \Pi_-=\sum_{n<\bar n}\dyad{n}, \qquad 
    \Pi_+=\sum_{n\ge \bar n}\dyad{n}.
\end{equation}
This implements a qubit flip conditioned on whether the oscillator energy lies above a threshold.

For an environment state $\sigma$, the induced channel on the probe is
\begin{equation}
    \Phi^{(\sigma,V)}_\mathcal{P}(\rho)
    = (1-p)\rho + p\, X\rho X~,
\end{equation}
with $p:=\Tr[\Pi_+\sigma]$. Thus, the interaction reduces to a bit-flip channel, fully determined by the weight $p$ of $\sigma$ above the energy threshold.

As in the previous examples, the probe dynamics preserves $\mathrm{STAB}_0$ whenever this mixture is Clifford-preserving, which occurs for
\begin{equation}
    p\in\{0,\tfrac{1}{2},1\}.
\end{equation}
This identifies the maximal induced free set
\begin{equation}
    S_{\max}=\left\{\sigma~\middle|~\Tr[\Pi_+\sigma]\in\{0,\tfrac{1}{2},1\}\right\}~,
\end{equation}
and, as before, one verifies that it is invariant under the back-action on $\mathcal{E}$, so the induced resource theory is strong.

To analyse harvesting, it is natural in this setting to restrict the environment to thermal states,
\begin{equation}
    \sigma_\beta=\frac{e^{-\beta H}}{Z}, \qquad 
    Z=\Tr(e^{-\beta H}),
\end{equation}
which introduces a physically motivated constraint on accessible states.

In this case, the control parameter becomes
\begin{equation}
    p(\beta)=\Tr(\Pi_+\sigma_\beta)
    = \frac{\sum_{n\ge \bar n} e^{-\beta \omega n}}{\sum_{m} e^{-\beta \omega m}}
    = e^{-\beta \bar n \omega}.
\end{equation}
Hence, the induced free condition $p\in\{0,\tfrac{1}{2},1\}$ translates into discrete values of the temperature,
\begin{equation}
    \beta\in\left\{0,\frac{\ln 2}{\bar n \omega},\infty\right\}.
\end{equation}

For a free input state $\rho=\dyad{0}$, the output is diagonal,
\begin{equation}
    \Phi^{(\sigma_\beta,V)}_\mathcal{P}(\rho)
    = \mathrm{diag}(1-p(\beta),\,p(\beta)),
\end{equation}
so the resource generated in the probe is entirely controlled by $p(\beta)$. In particular, resource generation occurs whenever $p(\beta)\notin\{0,\tfrac{1}{2},1\}$.

To interpret this, we equip the oscillator with a simplified resource theory of non-Gaussianity, where free states are Gaussian orbits of the vacuum together with the infinite-temperature state. Under the thermal constraint, this reduces to
\begin{equation}
    \mathcal{F}_\mathcal{E}^\beta=\{\dyad{0},\,\beta_\infty\}.
\end{equation}
In this regime one finds $\mathcal{F}_\mathcal{E}^\beta \subseteq S_{\max}$, so resource generation in the probe certifies non-Gaussianity in the environment, in line with the general witnessing mechanism.

However, the converse fails: there exists a finite temperature,
\begin{equation}
    \beta=\frac{\ln 2}{\bar n \omega},
\end{equation}
for which $p=\tfrac{1}{2}$ and no resource is generated in the probe, despite the environment being resourceful. Operationally, this corresponds to a perfectly balanced mixture of the two control branches, which washes out any detectable signature at the level of the probe.

\begin{figure}
    \centering
        \includegraphics[width=\columnwidth]{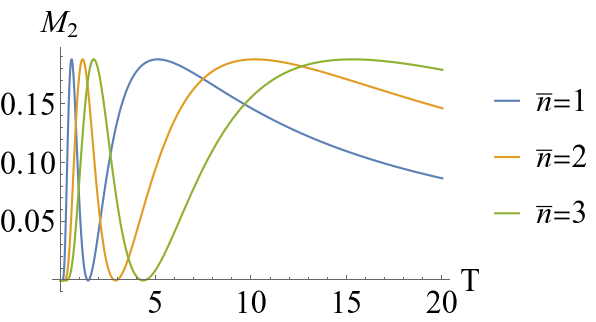}
        \caption{The SRE as a function of temperature. Here we set $\omega=1$ and plot different values for the cutoff $\bar{n}$.}
        \label{f.SRE_temp}
\end{figure}

\subsection{Qubit--oscillator Jaynes--Cummings}

We now consider a physically motivated interaction, namely the Jaynes--Cummings model describing a two-level atom coupled to a single bosonic mode. The probe $\mathcal{P}$ is the atom, while the field plays the role of the environment $\mathcal{E}$. The Hamiltonian is
\begin{equation}
    H = \Omega \dyad{1} + \omega a^\dagger a 
    + g(\dyad{0}{1}\otimes a^\dagger + \dyad{1}{0}\otimes a),
\end{equation}
and we work on resonance, $\Omega=\omega$, in the interaction picture.

We restrict the environment to states diagonal in the Fock basis,
\begin{equation}
    \sigma=\sum_n p_n \dyad{n},
\end{equation}
which is a natural constraint in many physical situations.

The unitary evolution $V_t=e^{-iH_{\mathrm{int}}t}$ induces a channel on the probe of the form
\begin{equation}
    \Phi^{(\sigma,V_t)}_\mathcal{P}(\rho)
    = \sum_{n} p_n \, \Phi_n^{(t)}(\rho),
\end{equation}
where each $\Phi_n^{(t)}$ corresponds to a Rabi oscillation with frequency $\sqrt{n}$ or $\sqrt{n+1}$. Explicitly, the Kraus operators are
\begin{align}
    K_{n1} &= \sqrt{p_n}
    \begin{bmatrix}
        \cos(gt\sqrt{n}) & 0 \\
        0 & \cos(gt\sqrt{n+1})
    \end{bmatrix}, \\
    K_{n2} &= \sqrt{p_n}
    \begin{bmatrix}
        0 & \sin(gt\sqrt{n+1}) \\
        \sin(gt\sqrt{n}) & 0
    \end{bmatrix}.
\end{align}

Unlike the previous engineered examples, the induced dynamics is now a superposition of oscillations at incommensurate frequencies. As a consequence, the set of environment states that preserve $\mathrm{STAB}_0$ is extremely restricted.

If one selects $S=\{\dyad{0}\}$, the channel reduces to a single Rabi oscillation. In this case, the probe dynamics is free only at discrete times
\begin{equation}
    t=\frac{m\pi}{2g}, \qquad m\in\mathbb{Z},
\end{equation}
where it becomes either the identity channel (even $m$) or an erasure map to $\dyad{0}$ (odd $m$). Both are free operations on the probe, so this choice defines a weak induced resource theory. It is strong only when the field is left invariant, which occurs for even $m$.

More generally, one finds that the maximal induced free set is $S_{\max}=\{\dyad{0}\}$, for these special times, and is empty for generic $t$. In other words, almost any interaction with the field generates a resource in the probe.

From the harvesting perspective, this has a clear implication: the protocol is not informative about the environment. Since $S_{\max}$ is either trivial or empty, one typically has $S_{\max} \subset\mathcal{F}_\mathcal{E}$, and resource generation in the probe does not certify the resourcefulness of the field. While this specific example implies an inability to learn about the environment state by inspecting if there is non-zero SRE in the probe post-interaction, it still signals that the natural interaction between a bosonic mode and an atom generates SRE in the state of the atom.

Operationally, what remains accessible is not a binary witness but a quantitative one: the amount and timescale of resource generation depend on the distribution $\{p_n\}$, and thus encode properties of the environment in a model-dependent way.

This limitation is specific to the choice of resource. If, instead, one considers the resource theory of coherence (with free states diagonal in the energy basis), the situation improves. For any diagonal $\sigma$ and any time $t$, the induced channel does not generate coherence in the probe. Hence, denoting by $S_{\max}^{\mathrm{coh}}$ the corresponding induced free set, one has $\mathcal{F}_\mathcal{E}^{\mathrm{coh}} \subseteq S_{\max}^{\mathrm{coh}}$, so coherence generation in the probe provides a faithful witness of coherence in the environment.

This example highlights a key point: harvesting is often possible for natural interactions, but its utility as a tool to learn information about the state of the environment crucially depends on the compatibility between the interaction and the chosen resource.

\subsection{Qubit--field with UDW interaction} \label{s.UDW_example}

Let us now discuss the interaction between a discrete quantum system $\mathcal{P}$ and a relativistic quantum field $\mathcal{E}$. Specifically, we consider $\mathcal{P}$ to be an Unruh--DeWitt detector with two energy levels separated by a gap $\Omega$, interacting for a finite time with a (1+1) dimensional massless scalar field. The interaction-picture time-evolution operator is
\begin{equation}
U = \mathcal{T} \left\{ \exp \left[-i \lambda\int\de^{D} x \, \Lambda(x) \, \mu(\tau) \otimes \phi(x) \right] \right\},
\end{equation}
where $\lambda \ll 1$ is the coupling strength, $\Lambda(x)$ is a spacetime smearing function, and the monopole operator is
\begin{equation}
\mu(\tau)=e^{-i\Omega\tau}\dyad{1}{0}+e^{+i\Omega\tau}\dyad{0}{1}.
\end{equation}
The scalar field operator is
\begin{equation}
\phi(x) = \int\frac{\de^d \mathbf{k}}{(2\pi)^d \sqrt{2|\mathbf{k}|}}\left( a_\mathbf{k} e^{-i k \cdot x} + a_\mathbf{k}^\dag e^{+i k \cdot x} \right),
\end{equation}
with $d=D-1$.

Expanding perturbatively up to second order in $\lambda$, the induced channel on the probe reads
\begin{widetext}
\ba
\nonumber
\Phi^{(\sigma, V)}\mathcal{P}(\rho) &=& \rho - i \lambda \int \de^D x ,\Lambda(x) [\mu(\tau)\rho - \rho \mu(\tau)] \braket{\phi(x)}_\sigma \\
\nonumber
&& - \frac{\lambda^2}{2} \int \de^D x \de^D x' \Lambda(x) \Lambda(x') \mathcal{T}\left[\mu(\tau) \mu(\tau') \rho + \rho \mu(\tau) \mu(\tau') \right] \braket{\mathcal{T}{\phi(x) \phi(x')}}_\sigma \\ \label{e.QFT_map}
&&+ \lambda^2 \int \de^D x \de^D x' \Lambda(x) \Lambda(x') \mu(\tau) \rho \mu(\tau') \braket{\phi(x') \phi(x)}_\sigma + \mathcal{O}(\lambda^3).
\label{e.probe_field_final_state_text}
\ea
\end{widetext}

We begin by analysing induced resources. Taking $\mathcal{F}_\mathcal{P} = \rm{STAB}_0$, one finds that perturbation theory generically maps stabilizer states to non-stabilizer states, as the free states $\rm{STAB}_0$ are not path-connected. The induced set $S$ is therefore empty. To obtain a non-trivial structure, we instead enlarge the free set to the convex stabilizer polytope
\begin{equation}
\rho = \tfrac{1}{2} (\mathbb{I} + \mathbf{r}\cdot \mathbf{X}), \quad |\mathbf{r}|_1 \leq 1,
\end{equation}
where $\abs{~\cdot~}_1$ denotes the 1-norm and $\mathbf{X}:= [X~ Y~ Z]^T$. We call this set of states $\rm{STAB}_\text{con}$. Even then, the map typically drives states outside this set. For instance, rewriting $\mu(\tau) = \cos(\Omega \tau) X + \sin(\Omega \tau) Y$, one finds contributions in the second line of Eq.~\eqref{e.probe_field_final_state_text} proportional to $\{Z,\rho\}$, which moves the state along the Bloch sphere in a way that generally exits the stabilizer polytope.

\begin{figure*}
     \centering
     \begin{subfigure}[b]{0.48\textwidth}
         \centering
         \includegraphics[width=\textwidth]{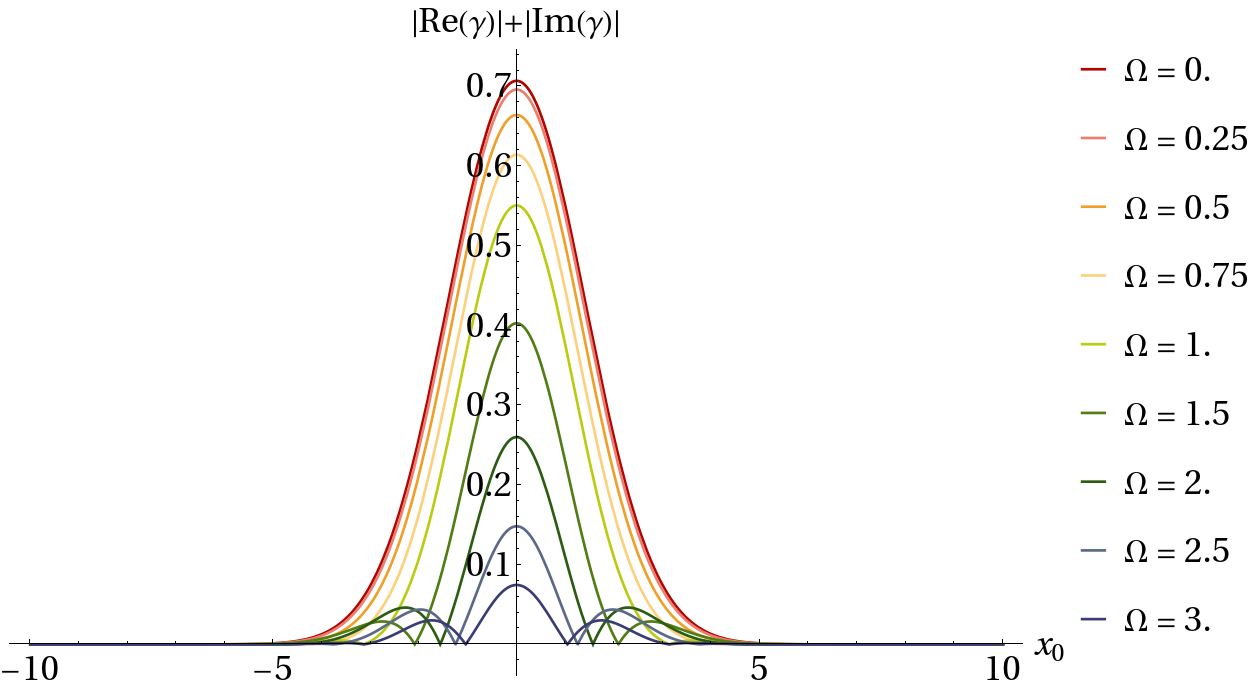}
         \caption{}
     \end{subfigure}
     \hfill
     \begin{subfigure}[b]{0.48\textwidth}
         \centering
         \includegraphics[width=\textwidth]{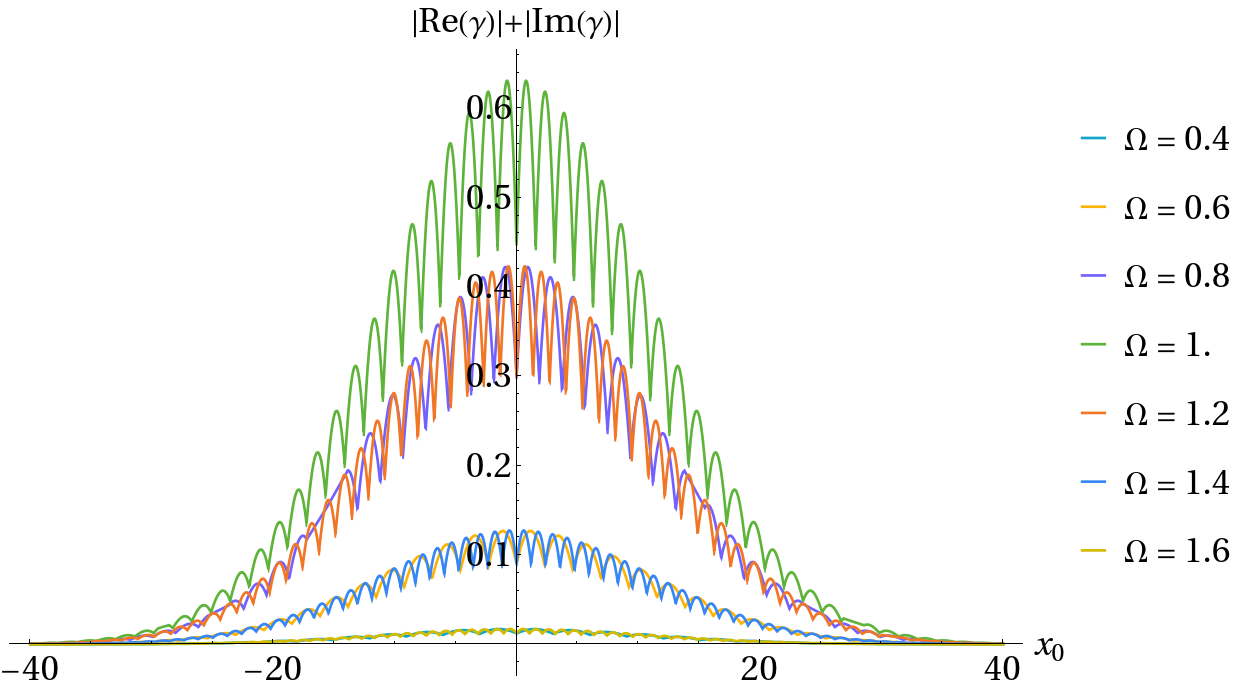}
         \caption{} \label{f.gamma_coh_state_b}
     \end{subfigure}
    \caption{The 1-norm of the coherence $\abs{\Re(\gamma)} + \abs{\Im(\gamma)}$ as a function of the localization peak of the coherent state $x_0$. For \textbf{(a)} $T=2, \sigma_0=2$ and $k_0=0$, while for \textbf{(b)} $T=10$, $\sigma_0=5$ and $k_0 = 1$. In \textbf{(b)}, one can see the resonance effect, as setting $\Omega=1$ yields the largest effect on the detector state.}
    \label{f.gamma_coh_state}
\end{figure*}

Thus, the interaction with a quantum field generically creates magic regardless of the initial field state. As in the Jaynes--Cummings case, this indicates that the induced resource theory reflects properties of the interaction itself rather than the resource content of the environment.

We now turn to resource harvesting. Fixing the initial probe state to $\rho_0=\dyad{0}$, the evolved state takes the form
\begin{equation}
\Phi^{(\sigma, V)}_\mathcal{P}(\dyad{0}) =
\begin{bmatrix}
1 - \lambda^2 q & \lambda \gamma \\
\lambda \gamma^* & \lambda^2 q
\end{bmatrix}
+ \mathcal{O}(\lambda^3),
\end{equation}
with
\begin{equation}
    \begin{split}
        q &= \int \de^D x \de^D x' \Lambda(x)\Lambda(x') e^{i\Omega(\tau-\tau')}\braket{\phi(x') \phi(x)}_\sigma, \\
\gamma &= i \int \de^D x \Lambda(x) e^{-i\Omega \tau} \braket{\phi(x)}_\sigma.
    \end{split}
    \label{e.gamma_ex_e}
\end{equation}
Requiring the state to remain inside $\rm{STAB}_\text{con}$ leads to the condition
\begin{equation}
|\Re \gamma| + |\Im \gamma| \leq \lambda q.
\end{equation}
This effectively enforces $\gamma \ll 1$, yielding the maximal induced free set
\begin{equation}
S_\text{max}^{(\dyad{0})} = \left\{ \sigma ~~\middle|~~ |\Re \gamma| + |\Im \gamma| \leq \lambda q \right\}.
\end{equation}
This set depends sensitively on the interaction profile $\Lambda(x)$.

If we define a resource theory on the field where free states satisfy $\braket{\phi(x)}_\sigma = 0$, then $\gamma=0$ and the inequality is always satisfied. Hence $\mathcal{F}_\mathcal{E} \subseteq S_\text{max}^{(\dyad{0})}$, so magic generation in the detector witnesses a non-vanishing one-point function in the field. Since the one-point function acts as an order parameter in symmetry breaking~\cite{Weinberg1995, peskin2018introduction, Sachdev1999}, this provides a direct operational probe of phase transitions.

To make this explicit, consider the field in a coherent state
\begin{equation}
\ket{\alpha} = \exp\left\{ \int \de^d \mathbf{k} \left[ \alpha(\mathbf{k}) a_\mathbf{k}^\dagger - \alpha^*(\mathbf{k}) a_\mathbf{k} \right] \right\} \ket{0},
\end{equation}
which has a non-zero one-point function. We use this to calculate the expression for $\gamma$. The details of the calculations together with motivation for the specific choice of state is presented in App.~\ref{a.alpha}. As depicted in Fig.~\ref{f.gamma_coh_state}, for Gaussian-localised wave packets, one finds that $\gamma$ increases when the field excitation overlaps with the detector region, and exhibits a resonance when the detector gap matches the characteristic momentum of the excitation.

This illustrates two key features: first, detectability depends on spacetime localization, not just on the presence of the resource; second, coherent states have no direct conceptual connection to the resource of magic, but they nonetheless induce magic in the detector. This second point highlights that induced resource theories classify how environmental states affect the probe and do not necessarily inherit the physical interpretation of the probe’s resource theory, unless specific settings are engineered or found in nature.

\section{Conclusions and outlook}
\label{sec:conclusions}
We examined the concept of quantum resource harvesting, intending to clarify under which conditions one is allowed to claim that harvesting of some resource in the probe implies the presence of some resource in the field.

Specifically, we have considered situations in which either the resource theory is not defined over the quantum system from which it is intended to be harvested, or physical constraints require the use of a non-free operation to perform the harvesting.

Much more generally, we extended this framework to settings where a well-understood probe system interacts with an environment, and we asked what we can learn about the resources of the environment when harvesting occurs. This question is particularly relevant when the interaction is physically motivated.

This has led us to introduce the concept of induced resource theory, namely, a resource theory for the environment defined in terms of the probe's resource theory. Moreover, depending on the nature of the interaction between the probe and the environment and on the properties of the quantum channel induced on the environment by the interaction with the probe, we have introduced two distinct types of induced resource theory: weak and strong. Our work not only clarifies the concept of quantum resource harvesting in non-standard settings but also paves the way for the investigation of harvesting protocols in non-ideal scenarios. 

The present framework points to several natural directions for further investigation. First, it could be generalised to a broad class of quantum resource theories, or combined with the approach of Ref.~\cite{DenerisEtAl26}, in which several resources are treated simultaneously.
Second, our results could be related to scenarios involving resource exchange between a probe and its environment, such as the trading between non-stabilizerness and athermality studied in Ref.~\cite{junior2025trading}. Furthermore, beyond detecting merely the presence of a resource in the probe, an important next step would be its quantitative characterisation, leading to the construction of induced resource monotones. A particularly promising direction is to explore connections between magic in the probe and one-point functions in quantum fields or spin chains. Establishing such a connection could provide a new approach to detecting and characterising phase transitions in these systems.

\section{Acknowledgements}
RN, NP, and EKV thank S. Nadal-Gisbert, S. Hirpara, S. Salomaa, and O. Veltheim for discussions and comments. RN acknowledges the financial support of the Vilho, Yrjö and Kalle Väisälä Foundation. NP acknowledges financial support from the Academy of Finland via the Centre of Excellence program (Project No. 336810 and Project No. 336814) during the initial stage of this work. NP and EKV acknowledge the financial support of the Research Council of Finland through the Finnish Quantum Flagship project (358878, UH). EKV is in part supported by the Research Council of Finland grant 1371600. A.H. and S.C. acknowledge support from  PNRR MUR project PE0000023-NQSTI and A.H. the PNRR MUR project CN 00000013-ICSC.

\bibliography{bib}

\onecolumngrid
\appendix

\section{The Stabilizer Rényi Entropy} \label{s.SRE}
Stabilizer R\'enyi Entropy ~\cite{Leone_2022} is a measure of how much the state of a quantum system is ``spread'' over the Pauli basis. For pure states, it is the {\em unique} computable monotone \cite{Leone_2024} for non-Clifford resources. 
Consider the probability distribution defined by:
\ba
\Xi_P(\ket{\psi})=\frac{\bra{\psi}P\ket{\psi}^2}{d},\;P\in\mathbb{P}_n,
\ea
where $d=2^n$ is the dimension of the Hilbert space.
The $\alpha$-SRE $M_\alpha(\ket{\psi})$ of an n qubits pure state $\ket{\psi}$ is defined as the $\alpha$ Rényi entropy of $\Xi_P$:
\ba
M_{\alpha}(\ket{\psi})=(1-\alpha)^{-1}\log\sum_{P\in\mathbb{P}_n}\Xi_P^\alpha(\ket{\psi})-\log d
\ea
In particular, the $2$-SRE $M_2$ for pure states can be written as:
\ba
\nonumber
M_2(\ket{\psi})&=& -\log\left[d\sum_{P\in \mathbb{P}_n} \Xi_P(\psi)^2\right]\\
\nonumber
&=&-\log \left( d^{-1} \Tr^4 (\psi P) \right)\\
&=&-\log d\Tr{Q\dyad{\psi}^{\otimes4}}
\ea
where $Q:=d^{-2}\sum_{P\in\pauli{n}}P^{\otimes4}$.

The 2-SRE can be extended to generic mixed states via:
\ba
\label{e.sre_mixed}
\tilde{M}_2(\rho)=M_2(\rho)-S_2(\rho),
\ea
where $M_2(\rho)=-\log d\Tr{Q\rho^{\ot4}}$ and $S_2(\rho)=-\log\Tr{\rho^2}$ is the 2-R\'enyi entropy of $\rho$.

The {\em non-local} non-stabilizerness is defined \cite{cao2024gravitationalbackreactionmagical} as $M^{NL} (\phi):= \min_{R=U_A\otimes U_B} M (R (\phi))$ for any non-stabilizer monotone $M$. It represents the non-stabilizer resources that cannot be erased by local unitary operations.

SRE can be generalised to qudit.
A generalization of the Pauli group is given by the Weyl-Heisenberg group \cite{Gross_2006}
\ba
\tilde{\mathcal{D}}^{(l)}_1=\left<\{D_{(p,q)}\}\right>=\{\tau^aD_{(p,q)}\}\quad \text{with}\quad a,p,q\in Z_l
\ea
i.e. the group generated by the $l^2$ operators defined as
\ba \label{whgen}
D_{(p,q)}=\tau^{-pq}Z^{p}X^{q} \quad \text{with}\quad \tau=e^{\frac{i}{l}\pi} \quad p,q\in Z_d 
\ea
and 
\ba
X\ket{j}=\ket{j\oplus_l 1} \quad
Z\ket{j}=\omega^j\ket{j} \quad \omega=e^{\frac{2\pi i}{l}}\,. \ea
The construction of Weyl-Heisenberg operators generalises straightforwardly for $n$ qudit systems via tensor product:
\ba
\tilde{\mathcal{D}}_n ^{(l)}=\tilde{\mathcal{D}}_1 ^{(l)}\otimes\underbrace{\dots}_n\otimes\tilde{\mathcal{D}}_1 ^{(l)} 
\ea

The SRE  have an analogous form when the Pauli decomposition is substituted with the Weyl-Heisenberg decomposition \cite{Wang2023}:
\ba 
\Xi_{D_{(p,q)}}(\ket{\psi}) & = & d^{-1}|{\bra{\psi}D_{(p,q)}\ket{\psi}}|^2 \\
M_{\alpha}(\ket{\psi}) & = & (1-\alpha)^{-1}\log \sum_{D_{(p,q)}\in \tilde{\mathcal{D}}_n} \Xi^{\alpha}_{D_{(p,q)}}(\ket{\psi})-\log d
\ea 

Most importantly, the analogy between qubits and qudits holds also in the case of mixed states, that is:
\ba
\label{e.d_SRE_mixed}
M_2(\rho)=-\log d^{-1}\sum_{p,q}\left|\Tr\left[D_{(p,q)}\rho\right]\right|^4-S_2(\rho)
\ea

\section{Details for Example \ref{s.UDW_example}} \label{a.alpha}
In this appendix, we show explicitly how to evaluate the $\gamma$ in Eq.~\eqref{e.gamma_ex_e} for a coherent state:
\begin{equation}\label{e.coh_state}
    \ket{\alpha} = \exp\left\{ \int \de^d \mathbf{k} \left[ \alpha(\mathbf{k}) a_\mathbf{k}^\dagger - \alpha^*(\mathbf{k}) a_\mathbf{k} \right] \right\} \ket{0}~,
\end{equation}
where $\alpha:\mathbb{R}^d \to \mathbb{C}$ denotes the momentum profile of the continuous multimode coherent state parameter. Coherent states physically correspond to field configurations with a Poisson distribution of particle numbers and are eigenstates of the annihilation operator $a_\mathbf{k}\ket{\alpha} = \alpha(\mathbf{k}) \ket{\alpha}$. The resulting one-point function is given by
\begin{equation}
    \braket{\phi(x)}_\alpha = \int \frac{\de^d \mathbf{k}}{(2\pi)^d \sqrt{2 \omega_\mathbf{k}}} \left[ \alpha(\mathbf{k})e^{-i k \cdot x} + \text{h.c} \right]~.
\end{equation}
Choosing a convenient spacelike foliation and splitting the spacetime smearing of the detector into its temporal and spatial part $\Lambda(x) = \chi(\tau) F(\mathbf{x})$, we obtain the expression
\begin{equation} \label{e.gamma_integral}
\begin{split}
    \gamma = i \int & \frac{ \de \tau \de^d \mathbf{k} \chi(\tau) }{(2\pi)^d \sqrt{2\omega_\mathbf{k}}} \bigg[ \alpha(\mathbf{k})\Tilde{F}_+(\mathbf{k}) e^{-i (\Omega + \omega_\mathbf{k}) \tau} + \alpha^*(\mathbf{k}) \Tilde{F}_-(\mathbf{k}) e^{-i ( \Omega - \omega_\mathbf{k}) \tau} \bigg]~,
\end{split}
\end{equation}
with
\begin{equation}
    \Tilde{F}_\pm (\mathbf{k}) := \int \de^d \mathbf{x} F(\mathbf{x}) e^{\pm i \mathbf{k}\cdot \mathbf{x}}~.
\end{equation}
Notably, the above expression for $\gamma$ depends on the overlap between the localization of the coherent state excitations and the localization of the detector--field interaction in momentum space.

For evaluation of the above, we pick a simple set of Gaussian localisations in a (1+1) dimensional spacetime and fix the scalar field to be massless $\omega_k = \abs{k}$. Namely, we set
\begin{align}
    \chi(\tau) =& e^{-\tau^2/2T^2} \\
    F(x) =& \delta(x) \quad \Rightarrow \quad \Tilde{F}_\pm (k) = 1 \\
    \alpha(k) =& \half \sqrt{\frac{\omega_k}{2}} \bigg[ e^{-\half \sigma_0^2 (k-k_0)^2}e^{-ikx_0} + e^{-\half \sigma_0^2 (k+k_0)^2}e^{+ikx_0} \bigg]~.
\end{align}
These choices correspond to a point-like detector localised at $x=0$ and a coherent state localised as a Gaussian wavepacket with peak intensity at $x=x_0$ for time $t=0$ and smeared with variance $\sigma_0^2$. The expression for $\alpha$ is motivated below. Specifically, we show it corresponds to a classical field configuration of a Gaussian packet moving in both directions with mean momentum $k_0$. 

The desired configuration of a (real) massless classical field has the following initial conditions ($t=0$):
\begin{align}\label{e.app_initial_conditions}
    \varphi_\text{cl}(x) &= \frac{1}{\sqrt{2\pi \sigma^2}} e^{-(x-x_0)^2/2\sigma_0^2}\cos(k_0 x) \\
    \pi_\text{cl}(x) &= 0~.
\end{align}
We could also include a constant amplitude for the field, modifying the expectation value of the excitations, but we set this to 1 for simplicity, as it does not change the behaviour of $\gamma$. This field configuration corresponds to a Gaussian wave packet with its peak at $x=x_0$, and moves in both directions as seen by its Fourier transform:
\begin{equation} \label{e.class_field_decomp}
    \Tilde{\varphi}_\text{cl}(k) = \half e^{-\half \sigma_0^2 (k-k_0)^2}e^{-ikx_0} + \half e^{-\half \sigma_0^2 (k+k_0)^2}e^{+ikx_0} \\
\end{equation}
The one-point functions of a coherent state describe its semi-classical field phase-space configuration as:
\begin{align} 
    \braket{\phi(t=0,x)}_\alpha = \varphi_\text{cl}(x) = \int \frac{\de k}{2\pi} \Tilde{\varphi}_\text{cl}(k) e^{ikx} \\
   \braket{\pi(t=0,x)}_\alpha = \pi_\text{cl}(x) = \int \frac{\de k}{2\pi} \Tilde{\pi}_\text{cl}(k) e^{ikx}~.
\end{align}
Noting that $a_k \ket{\alpha} = \alpha(k) \ket{\alpha}$, we equivalently obtain
\begin{align}
    \braket{\phi(t=0,x)}_\alpha &= \int \frac{\de k}{2\pi \sqrt{2\omega_k}} \left[ \alpha(k) + \alpha^*(-k) \right] e^{ikx} \\
    \braket{\pi(t=0,x)}_\alpha &= -i \int \frac{\de k}{2\pi \sqrt{2\omega_k}} \omega_k \left[ \alpha(k) - \alpha^*(-k) \right]e^{ikx} ~.
\end{align}
For consistency, we thus require
\begin{equation}
    \alpha(k) = \sqrt{\frac{\omega_k}{2}} \left( \Tilde{\varphi}_\text{cl} + \frac{i}{\omega_k} \Tilde{\pi}_\text{cl} \right)~.
\end{equation}
Inserting the initial conditions~\eqref{e.app_initial_conditions}, we obtain:
\begin{equation}
    \alpha(k) = \half \sqrt{\frac{\omega_k}{2}} \left[ e^{-\half \sigma_0^2 (k-k_0)^2}e^{-ikx_0} + e^{-\half \sigma_0^2 (k+k_0)^2}e^{+ikx_0} \right]~.
\end{equation}

For these particular choice of $\alpha$ and choice of smearing, we may evaluate the integral in the expression~\eqref{e.gamma_integral} numerically and obtain the results presented in Fig.~\ref{f.gamma_coh_state}.

\end{document}